\documentclass[journal]{IEEEtran}
\usepackage{amsfonts}
\usepackage{amsmath}
\newtheorem{lemma}{Lemma}

\ifCLASSINFOpdf
\else
\fi
\begin{document}
%
\title{The rates of convergence for generalized  entropy of  the normalized
sums of IID random variables}
%
%
%

\author{ Hongfei~Cui,~\IEEEmembership{}
       Jianqiang~Sun~\IEEEmembership{}
        and~Yiming~Ding ~\IEEEmembership{}

\thanks {Hongfei~Cui is with Wuhan Institute of Physics and Mathematics, The Chinese Academy of Sciences, Wuhan 430071, China
 (e-mail:
cui@wipm.ac.cn). Jianqiang~Sun is with Wuhan Institute of Physics
and Mathematics, The Chinese Academy of Sciences, Wuhan 430071,
China, and Graduate School of the Chinese Academy of Science
(e-mail: sunjianqiang09@mails.gucas.ac.cn).
  Yiming~Ding  is with Wuhan Institute of Physics and Mathematics,
  The Chinese Academy of Sciences, Wuhan 430071, China
  (e-mail:ding@wipm.ac.cn). This work was partially supported by National Basic Research
  Program of China (973 Program) Grant No. 2011CB707802.
 }}
\maketitle

\begin{abstract}

We consider the generalized differential entropy of normalized sums
of independent and identically distributed (IID) continuous random
variables. We prove that the R\'{e}nyi entropy and Tsallis entropy
of order $\alpha\ (\alpha>0)$ of the normalized sum of IID
continuous random variables with bounded moments are convergent to
the corresponding R\'{e}nyi entropy  and Tsallis entropy of the
Gaussian limit, and obtain sharp rates of convergence.
\end{abstract}

\begin{IEEEkeywords}
R\'{e}nyi entropy, Shannon entropy, Tsallis entropy, central limit
theorem, rate of convergence.
\end{IEEEkeywords}

{\bf MSC2000 Classification:} Primary 94A17  Secondary 62B10.

%
\IEEEpeerreviewmaketitle

\section{Introduction}
%
%
%
%
\IEEEPARstart{T}{he} Shannon entropy of a random variable $X$ with
density $f:\mathcal {R} \longrightarrow [0 , \infty)$ is defined as
$$
H(X)=-\int_{\mathcal {R}}f\log f
$$
provided that the integral make sense \cite{SH}. (We use $\log$ to
represent the natural logarithm throughout this paper).
It is interesting to study
the convergence of the normalized sums
$$S_n = \frac{1}{\sqrt n}\sum_{i=1}^{n}X_i$$
of independent copies $X$ to the Gaussian limit: the central limit
theorem for independent and identically distributed (IID) copies of
$X$. Without loss of generality, we suppose $\mathbb{E} (X)=0$ and $
\mathbb{E} (X^2)=1$. The standard Gaussian distribution is denoted
by $G$.

The idea of tracking the central limit theorem using Shannon entropy
goes back to Linnik \cite{LI}, who used it to give a particular
 proof of the central limit theorem. Barron \cite{BA} was
the first to prove a central limit theorem with convergence in the
Shannon entropy sense. He proved that $H(S_n)$ converges to
$H(G)=\log \sqrt{2\pi e}$ if $H(S_n)$ is finite for some $n>0$.
Notice that $H(S_n)$ is finite for some $n>0$ if  for some $s>0$,
$E|S_n|^{s}<\infty$, since
$$
H(S_n)\leq \frac{1}{s} \log [E|S_n|^{s} {(2\Gamma(\frac{1}{s}))}^s e
s^{1-s}],
$$
where $\Gamma$ is the Gamma function.  Artstein et al. \cite{ABBN2}
, Johnson  and Barron \cite{JB} obtained the rate of convergence
$$|H(S_n)-H(G)|=O(\frac{1}{n}),$$ provided the density of $X$
satisfies some analytical conditions \cite{J}. Moreover,  the
conclusion that $H(S_n)$ is increasing to $H(G)$ was obtained by
Artstein et al.\cite{ABBN1}, and some simpler proofs can be found in
Tulino and Verdu \cite{TV} and Madiman and Barron \cite{MB}.

R\'{e}nyi entropy  is a generalization of Shannon entropy \cite{RE}.
It is one of a family of functionals for
 quantifying the diversity, uncertainty or randomness of a system.
The R\'{e}nyi  entropy of order $\alpha$ $(\alpha\in \mathcal {R})$
is defined as
$$R_{\alpha}(X)= \frac{1}{1-\alpha}\log \int_\mathcal {R} {f^\alpha(x)} dx, \ \alpha \neq1 .$$
By L.Hoptital's rule, $R_{\alpha}(X) \to H(X)$ as $\alpha \to 1$.
The R\'{e}nyi entropy of order $2$, $R_2(X)$, is called Collision
entropy.

The R\'{e}nyi entropies are important in ecology and statistics as
indices of diversity. R\'{e}nyi  entropies  appear also in several
important contexts such as information theory, statistical
estimation, and quantum entanglement \cite{JV,LYG,ASW}. A class of
R\'{e}nyi entropy estimators for multidimensional densities are
given by Leonenko et al. \cite{LPS}.

Another generalization of Shannon entropy is Tsallis entropy
\cite{T}, which is defined as
$$T_{\alpha}(X)=\frac{1}{\alpha-1}(1-\int_\mathcal {R} f^{\alpha}(x)dx),\ \  \alpha\neq 1.$$
It is easy to see that $T_{\alpha}(X)\rightarrow H(X)$ as
$\alpha\rightarrow 1$.  Historically, this family of entropies was
derived by Havrda and Ch\'{a}rvat in 1967 \cite{HC}. The  R\'{e}nyi
entropy and Tsallis entropy with same order $\alpha$ are related by
\cite{NN}
$$
R_{\alpha}(X)=\frac{1}{1-\alpha} \log (1-(\alpha-1)T_{\alpha}(X)),
$$
which is a one-to-one correspondence between $R_{\alpha}(X)$ and
$T_{\alpha}(X)$. In fact, $R_{\alpha}(X)$ is a monotone increasing function of $T_{\alpha}(X)$,  because
$$
\frac{d R_{\alpha}(X)}{d
T_{\alpha}(X)}=\frac{1}{1-(\alpha-1)T_{\alpha}(X)}=\frac{1}{\int_\mathcal{R}f^{\alpha}(x)
dx}>0
$$
provided $\int_\mathcal{R}f^{\alpha}(x) dx<\infty$.

Parallel to the Shannon case, we consider the convergence of
$R_{\alpha}(S_n)$ and $T_{\alpha}(S_n)$ ($\alpha>0$), and
investigate the rates of convergence.

Our main results are the following Theorem. \vspace{0.5cm}

\textbf{Main Theorem} {\it Let $X_1,X_2,\ldots, X_n$ be independent
copies of a random variable $X$ with 
characteristic function $\varphi(t)$. 
Suppose the following hold:
\begin{enumerate}
\item $\mathbb{E}(|X|^k)$ is finite for $k=1,2,...$;
\item $|\varphi(t)|^\upsilon$ is integrable for some
$\upsilon \geq 1$.
\end{enumerate}
Then for any $\alpha>0$, we have
$$\lim_{n\to\infty}R_\alpha(S_n)=R_\alpha(G), \ \ \ \lim_{n\to\infty}T_\alpha(S_n)=T_\alpha(G).$$
Furthermore,
\begin{displaymath}|R_{\alpha}(S_n)-R_{\alpha}(G)|=
\begin{cases} O(n^{-\frac{1}{2}}) & 1<\alpha<\infty; \\
O(n^{-\frac{\alpha}{2}+\gamma}) & 0<\alpha \leq 1,
0<\gamma<\frac{\alpha}{2};
\end{cases}
\end{displaymath}
\begin{displaymath}|T_{\alpha}(S_n)-R_{\alpha}(G)|=
\begin{cases} O(n^{-\frac{1}{2}}) & 1<\alpha<\infty; \\
O(n^{-\frac{\alpha }{2}+\gamma}) & 0<\alpha \leq 1,
0<\gamma<\frac{\alpha}{2}.
\end{cases}
\end{displaymath}
\vspace{0.3cm}}

Although the  R\'{e}nyi entropy and Tsallis entropy can be defined
for any real number $\alpha$, we only consider the case $\alpha>0$.
In fact, for $\alpha \le 0$, one can check that
$R_{\alpha}(G)=\infty$ and $T_{\alpha}(G)=\infty$.

The bounded moment condition 1) is equivalent to
$\mathbb{E}(X^k)<+\infty$ for all positive integer $k$, which is
also equivalent to the fact that the characteristic function
$\varphi(t)$ admits derivatives of all orders at $t=0$. It is a
local condition imposed on $\varphi(t)$, while the integrability
condition 2) assumes a global property of $\varphi(t)$.

As we shall see in Lemma 5, the bounded moment condition 1) implies
the existence of $R_{\alpha}(S_n)$ and $T_{\alpha}(S_n)$ for all
$\alpha>0$. We shall obtain the rates of convergence via Feller's
expansion of densities. According to Feller's expansion \cite{FE} (Lemma 3),
the density of $S_n$ is convergent to the density of $G$ uniformly
with rate $o(n^{-1/2})$, where the integrability condition 2) is
assumed. Hence, the rates of convergence we obtained in Main Theorem
are sharp.

Since the moment condition is weaker than the Poincar\'{e}
inequality condition which was used in the Shannon case \cite{J},
the rate of convergence for Shannon entropy we obtained is
$O(n^{-1/2+\gamma})$ ($\gamma>0$ is small) rather than $O(n^{-1})$.

If the moment condition 1) in Main Theorem is replaced by
$\mathbb{E} |X|^3<\infty$, it is shown  that
$R_{\alpha}(S_n)\rightarrow R_{\alpha}(G)$ as $n\rightarrow+\infty$
for every $\alpha >1$, and rough rates of convergence are obtained
in \cite{CD}. \footnote{There is a small error in \cite{CD}, in
fact, the rates of convergence claimed in it should be divided by
$2$.}

\setcounter{equation}{0}
\section{Convergence of R\'{e}nyi entropy and Tsallis entropy}

Let $\{Y_n\}$ be a sequence of random variables with density
functions $\{p_n(x)\}$ and $Y$ be a random variable with density
function $p(x)$.  It is interesting to ask whether the R\'{e}nyi
entropy and Tsallis entropy of $\{Y_n\}$ of order $\alpha$
($\alpha>0$) are convergent the corresponding entropy of $Y$,
provided $Y_n \to Y$ in some sense. The following Theorem 1 claims
that if $\{p_n(x)\}$ is uniformly bounded, $\{p_n(x)\}$ is uniformly
convergent to $p(x)$, and the $L^{\alpha}$-norm of $\{p_n(x)\}$ is
uniform bounded for every $\alpha>0$, then the convergence results
hold. For discrete random variables, such kind of continuity is also
valid \cite{TMSA}.

 \vspace{0.3cm}

 \textbf{Theorem 1} {\it Let
$\{Y_n\}$ be a sequence of random variables with density functions
$\{p_n(x)\}$,  $Y$ be a random variable with density function
$p(x)$, and $\mathcal{A}$ be a subset of $\mathcal{R}$ with zero
Lebesgue measure. Suppose the following hold:
\begin{enumerate}

\item for any $\varepsilon>0$, there exists a positive integer  $N>0$ such that
$\underset{x\in\mathcal {R\backslash
A}}{\sup}\{|p_n(x)-p(x)|\}<\varepsilon$ for $n>N$;

\item there exists a finite number $M>0$ such that
 $p_n(x)\leq M$ uniformly in $x\in\mathcal {R\backslash A}$ and $n\in \mathbb{N}$;

\item for every $\alpha>0$, there exists a finite number
$M_{\alpha}>0$ such that $$\int_{\mathcal {R}}p_n^{\alpha}(x)dx\leq
M_{\alpha}$$ uniformly in $n\in \mathbb{N}$.
\end{enumerate}
Then for $\alpha>0$, we have
$$\lim_{n\to+\infty}R_\alpha(Y_n)=R_{\alpha}(Y), \ \ \ \ \ \lim_{ n\to+\infty}T_\alpha(Y_n)=T_{\alpha}(Y).$$
}

\textbf{Remark 1:}
\begin{enumerate}
\item It is interesting to note that one may use Theorem 1 to obtain the convergence results of  R\'{e}nyi entropy and Tsallis entropy
for random variables with correlations.

\item Theorem 1 claims the convergence of R\'{e}nyi entropy and Tsallis entropy  for all $\alpha>0$.
One can assume weaker conditions on the densities of $\{X_n\}$ to ensure the convergence for those $\alpha$ belong to some bounded and closed subset of $(0, \infty)$.

\item The condition 1) in Theorem 1 is equivalent to
the fact that  $\{p_n(x)\}$ converges to $p(x)$ uniformly on
$\mathcal {R\backslash A}$ as $n\to+\infty$. Combining with
condition 2) in Theorem 1, we know that $p(x)\leq M$. Moreover, in
the following Lemma 1 we also obtain that for every $\alpha>0$,
$$\int_{\mathcal {R}}p^{\alpha}(x)dx\leq
M_{\alpha}.$$

\item Since we are only interested with the asymptotic behavior of $\{p_n(x)\}$, it is enough to require the uniform boundedness of the $L^{\alpha}$-norm ($0<\alpha\leq \infty)$ of  $\{p_n(x)\}$ for $n>N_*$, where $N_*$ is a positive integer.

\end{enumerate}

\vspace{0.3cm}

We prepare two Lemmas which are important in the proof of Theorem 1.



\vspace{0.3cm}

\begin{lemma}\label{cov} Suppose that the conditions in Theorem 1
are
satisfied. Then for every $\alpha>0$, we have
\begin{equation}\label{conv-alpha}\lim_{n\to+\infty}\int_{\mathcal
{R}}|p^{\alpha}_n(x)-p^{\alpha}(x)|dx=0.\end{equation}\end{lemma}

\begin{proof} At first, we prove that for every $\alpha>0$,
$\int_{\mathcal {R}}p^{\alpha}(x)dx\leq M_{\alpha}.$ Since for every
$\alpha>0$ the function $f_{\alpha}(x)=x^\alpha$ is continuous
 and  $\{p_n(x)\}$ converges to $p(x)$ uniformly on
$\mathcal {R\backslash A}$ as $n\to+\infty$, we have that
$\{p^\alpha_n(x)\}$ converges to $p^\alpha(x)$ uniformly on
$\mathcal {R\backslash A}$ as $n\to+\infty$. By Fatou's Lemma and
condition 3 in Theorem 1, we conclude that for every $\alpha>0$,
$$\int_{\mathcal {R}}p^{\alpha}(x)dx=\int_{\mathcal {R}}\liminf_{n\to+\infty} p_n^{\alpha}(x)dx\leq \liminf_{n\to+\infty} \int_{\mathcal {R}}p_n^{\alpha}(x)dx\leq M_\alpha. $$

The proof of Lemma \ref{cov} is decomposed two cases: $\alpha>1$ and
$0<\alpha\leq1$.

For the case $\alpha>1$, using Lagrange mean value theorem, we have
\begin{eqnarray*}
\int _{\mathcal {R}}|p_n^{\alpha }(x)-p^{\alpha }(x)|dx
&=&\alpha\int_{\mathcal {R}}|p_n(x)-p(x)|\xi_n^{\alpha-1}(x)dx\\
&\leq&\alpha\underset{x\in \mathcal {R \backslash
A}}{\sup}|p_n(x)-p(x)|\int_{\mathcal {R}}\xi_n^{\alpha-1}(x)dx.
\end{eqnarray*}
where
$\min\{p_n(x),p(x)\}\leq\xi_n^{\alpha-1}(x)\leq\max\{p_n(x),p(x)\}$.
Next, we'll prove that $\int_{\mathcal {R}}\xi_n^{\alpha-1}(x)dx$ is
bounded.
\begin{eqnarray*}
\int_{\mathcal {R}}\xi_n^{\alpha-1}(x)dx&=&\int_{\{x:p_n(x)>p(x)\}}\xi_n^{\alpha-1}(x)dx\\
& &+\int_{\{x:p_n(x)\leq
p(x)\}}\xi_n^{\alpha-1}(x)dx\\
&\leq&\int_{\{x:p_n(x)>p(x)\}}p_n^{\alpha-1}(x)dx\\
& & +\int_{\{x:p_n(x)\leq
p(x)\}}p^{\alpha-1}(x)dx\\
&\leq&\int_{\mathcal {R}}p_n^{\alpha-1}(x)dx+\int_{\mathcal
{R}}p^{\alpha-1}(x)dx.
\end{eqnarray*}
Using the condition 3) in Theorem 1, we have $\int_{\mathcal
{R}}\xi_n^{\alpha-1}(x)dx\leq 2M_{\alpha-1}<+\infty$  for every
$\alpha\in(1,\infty)$. Hence we obtain that
\begin{equation}\label{conv-alpha-1}\int _{\mathcal {R}}|p_n^{\alpha }(x)-p^{\alpha }(x)|dx\leq 2\alpha M_{\alpha-1}\underset{x\in
\mathcal {R \backslash A}}{\sup}|p_n(x)-p(x)|.\end{equation}

Combining inequality (\ref{conv-alpha-1}) and the condition 1) in
Theorem 1, we obtain (\ref{conv-alpha}) for $\alpha>1$.

Now we consider the case $0<\alpha \leq1$.

The condition 3) in Theorem 1 indicates that for every $\gamma \in
(0, \frac{\alpha}{2})$, there exists an  $M_{2\gamma}>0$
 such that $\int_{\mathcal
{R}}p_{n}^{2\gamma}(x)dx \leq M_{2\gamma}$ and $\int_{\mathcal
{R}}p^{2\gamma}(x)dx\leq M_{2\gamma}$. It follows that
\begin{eqnarray}\label{up-bd}
& &\int_{\mathcal {R}}|p_n(x)+p(x)|^{2\gamma}dx \notag \\
&=&\int_{\{x:p_n(x)>p(x)\}}(p_n(x)+p(x))^{2\gamma}dx \notag\\
& & +\int_{\{x:p_n(x)\leq
p(x)\}}(p_n(x)+p(x))^{2\gamma}dx \notag \\
&\leq& \int_{\{x:p_n(x)>p(x)\}}(2p_n(x))^{2\gamma}dx \notag \\
& & +\int_{\{x:p_n(x)\leq
p(x)\}}(2p(x))^{2\gamma}dx \notag \\
&\leq&
\int_{\mathcal {R}}2^{2\gamma}p_n(x)^{2\gamma}dx+\int_{\mathcal {R}}2^{2\gamma}p(x)^{2\gamma}dx \notag \\
&\leq& 2^{2\gamma+1}M_{2\gamma}<\infty.
\end{eqnarray}

By inequality (\ref{up-bd}) and  the trivial inequality
$|b^{\alpha}-c^{\alpha}|\leq|b-c|^{\alpha}$, we have
\begin{eqnarray}\label{up-bd-1}
& & \int _{\mathcal {R}}|p_n^{\alpha}(x)-p^{\alpha}(x)|dx\notag\\
 &\leq& \int _{\mathcal {R}}|p_n(x)-p(x)|^{\alpha}dx\notag\\
&\leq& \int _{\mathcal {R}}|p_n(x)-p(x)|^{\alpha-2\gamma}|p_n(x)+p(x)|^{2\gamma}dx\notag\\
&\leq&\ \underset{x\in
\mathcal {R \backslash A}}{\sup}|p_n(x)-p(x)|^{\alpha-2\gamma}\int_{\mathcal {\mathcal {R}}}|p_n(x)+p(x)|^{2\gamma}dx\notag\\
&\leq& \ 2^{2\gamma+1}M_{2\gamma} \underset{x\in \mathcal {R
\backslash A}}{\sup}|p_n(x)-p(x)|^{\alpha-2\gamma}.
\end{eqnarray}

Combining inequality (\ref{up-bd-1}) and the condition 1) in Theorem
1, we obtain (\ref{conv-alpha}) for  $0<\alpha\leq 1$.
\end{proof}

\vspace{0.2cm}

\begin{lemma}\label{con}Suppose that the conditions in Theorem 1 be
satisfied. Then for any $\varepsilon
>0$ there exists a $\delta>0$ such that $$|T_{\alpha}(Y_n)-H(Y_n)|<\varepsilon
$$ for all $n\in\mathbb{N}$ and all $\alpha$ satisfying
$|\alpha-1|<\delta$.

Furthermore, for any $\varepsilon
>0$ there exists $\delta >0$ such that $$|T_{\alpha}(Y)-H(Y)|<\varepsilon $$ for all $\alpha$ satisfying
$|\alpha-1|<\delta$.
\end{lemma}

\begin{proof} The proof is given in Appendix.
\end{proof}

\vspace{0.5cm}

It is time to give the proof of Theorem 1.

\begin{proof} We consider two cases: $\alpha \neq 1$ and $\alpha=1$.

Suppose $\alpha\neq 1$. By Lemma \ref{cov}, for every $\alpha>0$
there exist $N>0$ such that $\int _{\mathcal {R}}p_n^{\alpha
}(x)dx\geq \frac{1}{2}\int _{\mathcal {R}}p^{\alpha }(x)dx$ for
$n>N$.
 Using the inequality $\mathrm{log}(1+x)<x$ ($x>0$), we have
\begin{eqnarray}\label{renyi}
& & |R_{\alpha}(Y_n)-R_{\alpha}(Y)| \notag\\
&=&\frac{1}{|1-\alpha
|}|\mathrm{log}\frac{\int_{\mathcal
{R}}p_n^{\alpha}(x)dx}{\int_{\mathcal
{R}}p^{\alpha}(x)dx}|\notag\\
&\leq&\max\left \{ \frac{\int_{\mathcal
{R}}p_n^{\alpha}(x)-p^{\alpha}(x)dx}{|1-\alpha|\int_{\mathcal
{R}}p^{\alpha}(x)dx},\frac{\int_{\mathcal
{R}}p^{\alpha}(x)-p_n^{\alpha}(x)dx}{|1-\alpha|\int_{\mathcal
{R}}p_n^{\alpha}(x)dx}
\right \}\notag\\
&\leq&\max\left \{ \frac{\int_{\mathcal
{R}}p_n^{\alpha}(x)-p^{\alpha}(x)dx}{|1-\alpha|\int_{\mathcal
{R}}p^{\alpha}(x)dx},\frac{\int_{\mathcal
{R}}p^{\alpha}(x)-p_n^{\alpha}(x)dx}{|1-\alpha|\frac{1}{2}\int_{\mathcal
{R}}p^{\alpha}(x)dx}
\right \}\notag\\
&\leq&\frac{ 2\int_{\mathcal {R}}\left
|p^{\alpha}(x)-p_n^{\alpha}(x)\right |dx}{|1-\alpha|\int_{\mathcal
{R}}p^{\alpha}(x)dx}.
\end{eqnarray}
Combining (\ref{renyi})  and Lemma \ref{cov}, we obtain that for
every $\alpha>0$ and $\alpha\neq1$,
$$\lim_{n\to \infty}|R_{\alpha}(Y_n)-R_{\alpha}(Y)|=0.$$

It is obvious that
\begin{eqnarray}\label{tsallis}
& &|T_{\alpha}(Y_n)-T_{\alpha}(Y)|\notag\\
&=&|\frac{1}{\alpha-1}(1-\int_{\mathcal {R}}
p_n^{\alpha}(x)dx)-\frac{1}{\alpha-1}(1-\int_{\mathcal {R}}
p^{\alpha}(x)dx)|\notag\\
&=&|\frac{1}{\alpha-1}\int_{\mathcal {R}}
p_n^{\alpha}(x)-p^{\alpha}(x)dx|\notag\\
&\leq& \frac{1}{|\alpha-1|}\int_{\mathcal
{R}}|p_n^{\alpha}(x)-p^{\alpha}(x)|dx.
\end{eqnarray}
We obtain that  for every $\alpha>0$ and $\alpha\neq 1$,
\begin{equation}\label{convergence}\lim_{n\to \infty}|T_{\alpha}(Y_n)-T_{\alpha}(Y)|=0.\end{equation}

Now we consider the case $\alpha=1$.

Remember that for any random variable $X$, $R_1(X)=T_1(X)=H(X)$.
By triangular inequality,
\begin{eqnarray}\label{entropy}
& &|H(Y_n)-H(Y)|\notag \\
&\leq&|H(Y_n)-T_{\alpha_0}(Y_n)|+|T_{\alpha_0}(Y_n)-T_{\alpha_0}(Y)|\notag\\
& & +|T_{\alpha_0}(Y)-H(Y)|.
\end{eqnarray}
Given any $\varepsilon
>0$,  there exists $\delta_1 >0$
such that $|H(Y_n)-T_{\alpha_0}(Y_n)|<\frac{\varepsilon}{3}$, for
all $n\in\mathbb{N}$ and all points $\alpha$ satisfying
$|\alpha_0-1|<\delta_1$ by Lemma \ref{con}. Similarly, there exists
$\delta_2 >0$  such that $|T_{\alpha_0}(Y)-H(Y)| \leq
\frac{\varepsilon}{3} $ all points $\alpha$ satisfying
$|\alpha_0-1|<\delta_2$.

By (\ref{convergence}), we know that for every  $\alpha_0\neq 1$ satisfying
$|\alpha_0-1|<\min\{\delta_1,\delta_2\}$, there exists an $N>0$ such
that $|T_{\alpha_0 }(Y_{n})-T_{\alpha_0}(Y)|\leq
\frac{\varepsilon}{3}$ for all $n>N$.

According to inequality (\ref{entropy}), we conclude that for any $\varepsilon>0$, there
exists an $N$ such that $ |H(Y_n)-H(Y)|\leq\varepsilon $ for $n \ge
N$.

Therefore, $H(Y_n)$ converges to $H(Y)$.

\end{proof}

\section{Proof of Main Theorem}

\begin{lemma}\label{expansions} (Expansions for densities, Feller \cite{FE}) \it{Let $X_1,X_2,\ldots ,X_n$ be independent copies of a random variable
$X$ with characteristic function $\varphi(t)$, $f_n$ be the density
function of normalized sum $S_n$, and $g$ be the density of the
standard Gaussian
   distribution $G$. Suppose that $\mathbb{E}|X|^3=\rho < \infty $, and that ${|\varphi|}^v $ is integrable for some $v > 1$. Then as $n \rightarrow \infty$
   \begin{eqnarray}
   f_n(x)- g(x)- \frac{\rho}{6\sqrt{n}} (x^3-3x)g(x)= o(\frac{1}{\sqrt{n}})
   \end{eqnarray}
   uniformly in $x$. }
\end{lemma}

\vspace{0.3cm}

\textbf{Remark 2:} Since $g(x)$ is the density function of the
standard Gaussian
   distribution $G$, $g(x)\leq \frac{1}{\sqrt{2 \pi}}<1$. For every $\alpha>0$, there exists an
$M'_{\alpha}>0$ such that $\int_\mathcal {R}g^\alpha(x)dx\leq
M'_\alpha$. According to Lemma \ref{expansions}, we have
\begin{eqnarray}\label{feller} \underset{x\in\mathcal
{R}}{\sup}\{|f_n(x)-g(x)|\}=o(n^{-\frac{1}{2}}).
\end{eqnarray}
Using (\ref{feller}), we  have that there exists an $N>0$ such that
$f_n(x)\leq 1$ uniformly for $n>N$ and $ x\in \mathcal {R}$. Hence,
without loss of generality, we can suppose that $f_n(x)\leq 1$
uniformly in $n\in\mathbb{N}$ and $ x\in \mathcal {R}$.

\vspace{0.2cm}

Next, we'll prove that $f_n(x)$ satisfy the condition 3) in Theorem
1.

\begin{lemma} \label{ineqal1} (Theorem 2.10 \cite{P}) {\it Let $X_1,X_2,\ldots ,X_n$ be independent random
variables with zero means, and let $k\geq2$ and
$Z_n=X_1+X_2+\ldots+X_n$. Then
$$\mathbb{E}|Z_n|^k=C(k)n^{\frac{k}{2}-1}\sum_{i=1}^{n}\mathbb{E}|X_i|^k,$$
where $C(k)$ is a positive constant depending only on $k$.}
\end{lemma}

\vspace{0.2cm}

\begin{lemma}\label{bounded}{\it  Let $X_1,X_2,\ldots ,X_n$ be independent copies of a random
variable $X$ with probability density $f$ and characteristic
function $\varphi(t)$. $f_n(x)$ be the density of the normlized sum
$S_n$. $\mathbb{E}(X)=0$, and $\mathbb{E}(|X|^k)=\rho_k$, where
$k=1,2,...$ and $\rho_k$ is finite. Then for each positive $\alpha$,
there exists an $M''_{\alpha }>0$ such that $$\int _{R}f_{n}^{\alpha
}(x)dx\leq M''_{\alpha }.$$}\end{lemma}

\begin{proof} We distinguish three cases: $\alpha=1$, $\alpha>1$ and
$0<\alpha<1$.

The case $\alpha=1$ is obvious because $\{f_n\}$ are densities. If
$\alpha>1$, noting that $\underset{x\in\mathcal
{R},n\in\mathbb{N}}{\mathrm{sup}f_n(x)}\leq 1$ in Remark 2 we have

$$
\int_{R}f_{n}^{\alpha}(x)dx =\int_{R}f_{n}^{\alpha-1}(x)f_n(x)dx
\leq 1.
$$

If $0<\alpha<1$, one can choose a positive integer $k$ such that
$\frac{1}{k+1}<\alpha$.  From Lemma \ref{ineqal1}, we obtain that
\begin{eqnarray*}
\mathbb{E}|S_n|^k&=&\mathbb{E}\left |(X_1+X_2+...+X_n)/n^\frac{1}{2} \right |^k\\
&=&\mathbb{E}|Z_n|^k/n^\frac{k}{2}\\
&\leq&C(k)n^{\frac{k}{2}-1}(\sum_{i=1}^{n}|X_i|^k)/n^\frac{k}{2}\\
&=&C(k)\mathbb{E}|X|^k\\
&:=&\rho'_k<\infty.
\end{eqnarray*}

On the other hand, noting that $\underset{x\in\mathcal
{R},n\in\mathbb{N}}{\mathrm{sup}f_n(x)}\leq 1$ in Remark 2 we have
that
\begin{eqnarray*}
& & \int_{\mathcal {R}}f_{n}^{\alpha}(x)dx\\
&=&\int_{-1}^{1}f_{n}^{\alpha} (x)dx+\int_{|x|\geq1}
f_{n}^{\alpha}(x)dx\\
&=&\int_{-1}^{1}f_{n}^{\alpha}(x)dx+\int_{|x|\geq1}
f_{n}^{\alpha}(x)|x|^{k\alpha}|x|^{-k\alpha}dx\\
&\leq&2+(\int_{|x|\geq1}f_{n}(x)|x|^kdx)^{\alpha}(\int_{|x|\geq1}|x|^{\frac{-k\alpha}{1-\alpha}}dx)^{1-{\alpha}}\\
&\leq&2+\rho_k'^{\alpha}(\int_{|x|\geq1}|x|^{-\frac{k\alpha}{1-\alpha}}dx)^{1-{\alpha}},
\end{eqnarray*}
where the first inequality follows from  H\"older inequality.

Since$\int_{|x|\geq1}|x|^{-\frac{k\alpha}{1-\alpha}}dx$ is finite
for $\frac{1}{k+1}<\alpha<1$, one can find a positive constant
$M''_{\alpha}$ independent of $n$ such that
$$\int_{\mathcal {R}}f_{n}^{\alpha}(x)dx<M''_{\alpha}.$$

\end{proof}

\vspace{0.2cm}

\begin{lemma}\label{ineqal2}{\it  If $0\leq y, x\leq 1$, then for every $\gamma \in
(0,\frac{1}{2})$, there exists a $Q_{\gamma}>0$ such that}
$$Q_{\gamma}|x^{1-\gamma}-y^{1-\gamma}|\geq|x^{\gamma}y^{1-\gamma}\log{x}-y^{\gamma}x^{1-\gamma}\log{y}|.$$
\end{lemma}

\begin{proof} The proof is given in Appendix.\end{proof}

\vspace{0.5cm}

\textbf{Proof of Main Theorem.}

\begin{proof}
From Remark 2, we know that $f_n(x)$ satisfies the condition 2) in
Theorem 1, and $f_n(x),g(x)$ satisfy the condition 1) in Theorem 1.
By Lemma \ref{bounded} and letting
$M_{\alpha}=\max\{M'_{\alpha},M''_{\alpha}\}$,  we obtain that
 for every $\alpha>0$  $\int_\mathcal {R}f_n^{\alpha}(x)dx,\int_\mathcal {R}g^{\alpha}(x)dx\leq M_{\alpha}$ uniformly in $n\in\mathbb{N}$ and  $\{f_n(x)\}$ satisfies the condition 3) in Theorem
 1.
 Hence from Theorem 1 we obtain that
$$\lim_{n\to\infty}R_\alpha(S_n)=R_\alpha(G), \ \ \ \lim_{n\to\infty}T_\alpha(S_n)=T_\alpha(G).$$

Next, we'll study the rates of convergence for $R_\alpha(S_n)$ and
$T_\alpha(S_n)$.

At first, we consider the case $\alpha>0$ and $\alpha\neq 1$.

Using the inequality $\log (1+x)<x\ (x>0)$, and inequalities
(\ref{renyi}), (\ref{conv-alpha-1}), (\ref{up-bd-1}), we have
\begin{eqnarray}\label{renyi-1}
&& |R_{\alpha}(S_n)-R_{\alpha}(G)|\notag\\
&\leq&\frac{ 2\int_{\mathcal {R}}\left
|f^{\alpha}_n(x)-g^{\alpha}(x)\right |dx}{|1-\alpha|\int_{\mathcal
{R}}f^{\alpha}_n(x)dx}\notag\\
&\leq&\begin{cases} \frac{4\alpha
M_{\alpha-1}\underset{x\in\mathcal{R}}{\sup}|f_n(x)-g(x)|}{\int_{\mathcal
{R}}p^{\alpha}(x)dx(\alpha-1)}&
\alpha\in(1,+\infty);\\
\frac{M_{2\gamma}\underset{x\in \mathcal
{R}}{\sup}|f_n(x)-g(x)|^{\alpha-2\gamma}}{2^{-2\gamma-2}\int_{\mathcal
{R}}p^{\alpha}(x)dx(1-\alpha)}& \alpha\in(0,1),\gamma
\in(0,\frac{\alpha }{2}).
\end{cases}\notag\\
\end{eqnarray}
Combining (\ref{feller}) and  (\ref{renyi-1}), we obtain
\begin{displaymath}|R_{\alpha}(S_n)-R_{\alpha}(G)|=
\begin{cases} O(n^{-\frac{1}{2}}) & \alpha\in(1,+\infty); \\O(n^{-\frac{\alpha
}{2}+\gamma}) & \alpha\in(0,1),\gamma \in (0,\frac{\alpha }{2}).
\end{cases}
\end{displaymath}

On the other hand, using the inequalities (\ref{tsallis}),
(\ref{conv-alpha-1}), (\ref{up-bd-1}), we have
\begin{eqnarray}\label{tsallis-1}
& &|T_{\alpha}(S_n)-T_{\alpha}(G)|\notag\\
&=&|\frac{1}{\alpha-1}\int_{\mathcal {R}}
f_n^{\alpha}(x)-g^{\alpha}(x)dx|\notag\\
&\leq& |\frac{1}{\alpha-1}|\int_{\mathcal
{R}}|f_n^{\alpha}(x)-g^{\alpha}(x)|dx\notag\\
&\leq&\begin{cases}  \frac{2\alpha M_{\alpha-1}\underset{x\in
\mathcal {R}}{\sup}|f_n(x)-g(x)|}{\alpha-1} & \alpha\in(1,+\infty);
\\\frac{M_{2\gamma}\underset{x\in
\mathcal
{R}}{\sup}|p_n(x)-p(x)|^{\alpha-2\gamma}}{2^{-2\gamma-1}(1-\alpha)}
& \alpha\in(0,1),\gamma \in (0,\frac{\alpha }{2}).
\end{cases}\notag\\
\end{eqnarray}
Combining  (\ref{feller}) and  (\ref{tsallis-1}), we obtain
\begin{displaymath}|T_{\alpha}(S_n)-T_{\alpha}(G)|=
\begin{cases} O(n^{-\frac{1}{2}}) & \alpha\in(1,+\infty); \\O(n^{-\frac{\alpha
}{2}+\gamma}) & \alpha\in(0,1),\gamma \in (0,\frac{\alpha }{2}).
\end{cases}
\end{displaymath}

Now we investigate the case $\alpha=1$.

Observe that for any random variable $X$,$$R_1(X)=T_1(X)=H(X).$$ In
what follows we show the following
\begin{equation}\label{rate for 1}|H(S_n)-H(G)|=O(n^{-\frac{1}{2}+\gamma}),\ \ for\ \gamma\in(0,\frac{1}{2}).\end{equation}

For $\gamma\in(0,\frac{1}{2})$, we have
\begin{eqnarray}\label{ineq}
& &n^{\frac{1}{2}-\gamma}|H(S_n)-H(G)|\notag\\
 &=&n^{\frac{1}{2}-\gamma}\left | \int
_\mathcal {R}f_n(x)\mathrm{log}f_n(x)dx-\int
_\mathcal {R}g(x)\mathrm{log}g(x)dx \right |\notag\\
 &=&n^{\frac{1}{2}-\gamma}|\int _\mathcal {R}f_n(x)\log{f_n(x)}-f_n^{\gamma}(x)g^{1-\gamma}(x)\log{f_n(x)}dx\notag\\
 &&+\int_\mathcal {R}f_n^{\gamma}(x)g^{1-\gamma}(x)\log{f_n(x)}-f_n^{1-\gamma}(x)g^{\gamma}(x)\log{g(x)}dx\notag\\
 &&+\int_\mathcal {R}f_n^{1-\gamma}(x)g^{\gamma}(x)\mathrm{log}g(x)-g(x)\mathrm{log}g(x)dx|\notag\\
&\leq&n^{\frac{1}{2}-\gamma}(|J_1(n)|+|J_2(n)|+|J_3(n)|),
\end{eqnarray}
where
$$J_1(n)=\int_\mathcal{R}f_n(x)\log{f_n(x)}-f_n^{\gamma}(x)g^{1-\gamma}(x)\log{f_n(x)}dx,$$
$$J_2(n)=\int_\mathcal{R}f_n^{\gamma}(x)g^{1-\gamma}(x)\log{f_n(x)}-f_n^{1-\gamma}(x)g^{\gamma}(x)\log{g(x)}dx,$$
$$J_3(n)=\int_\mathcal{R}f_n^{1-\gamma}(x)g^{\gamma}(x)\log
g(x)-g(x)\log g(x)dx.$$

By Lemma \ref{expansions}, for $\varepsilon=\frac{1}{2}$, there
exists an $N_1>0$, such that for each $n>N_1$, we have
\begin{eqnarray*}
& &n^{\frac{1}{2}}|f_n(x)-g(x)|\\
&\leq& \max\{|\frac{\rho}{6}(x^3-3x)g(x)+\frac{1}{2}|,|\frac{\rho}{6}(x^3-3x)g(x)-\frac{1}{2}|\}\\
& :=&m(x)\leq \max_{x\in\mathcal {R}}\{m(x)\}:=C_1.
\end{eqnarray*}
 Therefore,  we have
\begin{eqnarray*}
& &n^{\frac{1}{2}-\gamma}J_1(n) \\
&=&n^{\frac{1}{2}-\gamma}|\int _\mathcal {R}f_n(x)\log{f_n(x)}dx\\
& &-\int_\mathcal {R}f_n^{\gamma}(x)g^{1-\gamma}(x)\log{f_n(x)}dx|\\
&\leq&n^{\frac{1}{2}-\gamma}\int_\mathcal {R}f_n^{\gamma}(x)|f_n^{1-\gamma}(x)-g^{1-\gamma}(x)||\mathrm{log}f_n(x)|dx\\
&\leq&n^{-\frac{\gamma}{2}}\int_\mathcal {R}f_n^{\gamma}(x)n^{\frac{1-\gamma}{2}}|f_n(x)-g(x)|^{1-\gamma}|\mathrm{log}f_n(x)|dx\\
&\leq&\int_\mathcal {R}f_n^{\gamma}(x)(m(x))^{1-\gamma}|\mathrm{log}f_n(x)|dx\\
&\leq&C_1^{1-\gamma}\int_\mathcal
{R}f_n^{\frac{\gamma}{2}}(x)|f_n^{\frac{\gamma}{2}}(x)\mathrm{log}f_n(x)|dx.
\end{eqnarray*}

By Lemma \ref{bounded}, there exists $M_{\frac{\gamma}{2}}>0$ such
that
 $\int_Rf_n^{\frac{\gamma}{2}}(x)dx\leq M_{\frac{\gamma}{2}}$. On the other hand, noting that $f_n(x)\leq1$ uniformly in $n\in\mathbb{N}, x\in\mathcal {R}$ in Remark 2, hence when $n>N_1$, $n^{\gamma-1/2}J_1(n)\leq
 M_{\frac{\gamma}{2}}C_1^{1-\gamma}C_2<+\infty$, where $C_2=\underset{x\in\mathcal {R},n>N_1}{\mathrm{max}}|f_n^{\frac{\gamma}{2}}(x)\mathrm{log}f_n(x)|$.

Using  similar arguments, we can obtain that the there exist $N_2$,
$C_3>0$ such that $$n^{\frac{1}{2}-\gamma}J_3(n) \leq C_3\ \ \ \ for
\ \ n>N_2.$$

According to Lemma \ref{ineqal2}, for every
$\gamma\in(0,\frac{1}{2})$, there exists a $Q_{\gamma}>0$ such that
$Q_{\gamma}|f_n^{1-\gamma}(x)-g^{1-\gamma}(x)|\geq|f_n^{\gamma}(x)g^{1-\gamma}(x)\log{f_n(x)}-g^{\gamma}(x)f_n^{1-\gamma}(x)\log{g(x)}|.$

It follows that
\begin{eqnarray*}
n^{\frac{1}{2}-\gamma}J_2(n)&=&n^{\frac{1}{2}-\gamma}|\int _\mathcal {R}f_n^{\gamma}(x)g^{1-\gamma}(x)\log{f_n(x)}dx\\
& &-\int_\mathcal {R}f_n^{1-\gamma}(x)g^{\gamma}(x)\log{g(x)}dx|\\
&\leq&n^{\frac{1}{2}-\gamma} Q_{\gamma}\int_\mathcal
{R}|f_n^{1-\gamma}(x)-g^{1-\gamma}(x)|dx.
\end{eqnarray*}
Combining inequalities (\ref{conv-alpha-1}), (\ref{up-bd-1}) and
equality (\ref{feller}), we have
$$\int_\mathcal
{R}|f_n^{1-\gamma}(x)-g^{1-\gamma}(x)|dx=O(n^{-\frac{1}{2}+\gamma}).$$
Hence, there exists a $C_4>0$ such that
$n^{\frac{1}{2}-\gamma}J_2(n)\leq C_4$ for $n>N_3$.

From the above discussion and inequality (\ref{ineq}), there exists
constant $C:=M_{\frac{\gamma}{2}}C_1C_2+C_3+C_4$ such that
$$n^{\frac{1}{2}-\gamma}\left | \int_\mathcal
{R}f_n(x)\mathrm{log}f_n(x)dx-\int_\mathcal
{R}g(x)\mathrm{log}g(x)dx\right |\leq C,$$  for $n>\max\{N_1, N_2,
N_3\}$.

Hence, (\ref{rate for 1}) is true.

\end{proof}

\vspace{0.2cm}

\section{conclusion and discussion}

We show the convergence of the normalized sum of IID continuous
random variables with bounded moments of all order in the sense of
R\'{e}nyi entropy and Tsallis entropy, and obtain sharp rates of
convergence. By using Feller's expansion and detailed analytical
properties of the corresponding densities, we estimate the R\'{e}nyi
entropy and Tsallis entropy directly.  The main difficulty lies in
the case of Shannon entropy, both on the convergence and rate of
convergence, because $\alpha=1$ is a singularity of $R_{\alpha}$ and
$T_{\alpha}$. We circumvent it by obtaining some uniform estimations
near $\alpha=1$. Compared with the previous proof for Shannon
entropy, our proof is more direct,  can be generalized to random
vectors in higher dimension, and may be used to consider the
convergence of normalized sum of dependent random variables. The
bounded moment condition we used is weaker than the Poincar\'{e}
constant condition in Shannon case \cite{J}. As a result, the rates
of convergence is slower than $O(\frac{1}{n})$ in Shannon case. It
is interesting to consider the convergence and rates of convergence
for R\'{e}nyi divergence  as R\'{e}nyi raised in \cite{RE}.

\section{Appendix}

{\bf Proof of Lemma \ref{con}.}

\begin{proof}Observing that $$\frac{\partial(1-\int_\mathcal {R}p_n^{t}(x)dx)}{\partial t}=-\int_\mathcal {R}p_n^t(x)\log p_n(x)dx$$ and
$\left (1-\int_\mathcal {R}p_n^{t}(x)dx \right )|_{t=1}=0$, we have
that
\begin{equation}\label{eq}
1-\int_\mathcal {R}p_n^{\alpha}(x)dx=-\int_1^{\alpha}\int_\mathcal
{R}p_n^t(x)\log p_n(x)dxdt.
\end{equation}
From equality (\ref{eq}) we have that
\begin{eqnarray*}
&\quad&T_{\alpha}(Y_n)-H(Y_n)\\
&=&\frac{1-\int_\mathcal {R}p_n^{\alpha}(x)dx}{\alpha-1
}+\int_\mathcal {R}p_n(x)\mathrm{log}p_n(x)dx \\
&=&\int_1^{\alpha}\frac{\int_\mathcal
{R}p_n(x)\mathrm{log}p_n(x)dx-\int_\mathcal {R}p_n^t(x)\log
p_n(x)dx}{\alpha-1}dt\\
&=&\int_{1}^{\alpha }\frac{J_n(t)}{(\alpha-1)}dt,
\end{eqnarray*}
where
$$J_n(t)=\int_\mathcal {R}p_n(x)\mathrm{log}p_n(x)dx-\int_\mathcal
{R}p_n^t(x)\log p_n(x)dx.$$ It follows that
\begin{eqnarray}\label{R_S}
|R_{\alpha}(Y_n)-H(Y_n)| \leq\frac{1}{|1-\alpha |}\int_{1}^{\alpha}
|1-t| \left |\frac{J_n(t)}{1-t}\right |dt.
\end{eqnarray}

Next, we'll prove $\left |J_n(t)/(1-t)\right |$ is bounded uniformly
in $t\in [\frac{3}{4},\frac{3}{2}]$ and $n\in\mathbb{N}$.

Using Lagrange mean value Theorem, for fixed $n$ and $x$, there
exists $\xi_n(x)$ which is between $t$ and $1$ and in
$[\frac{3}{4},\frac{3}{2}]$ such that
$p_n^t(x)-p_n(x)=p_n^{\xi_n(x)}(x)\log p_n(x)(t-1)$.  We have
\begin{eqnarray*}
&\quad& |J_n(t)/(1-t)|\\
&=&\left|\frac{\int_\mathcal
{R}(p_n(x)-p_n^t(x))\mathrm{log}p_n(x)dx}{1-t}\right|\\
&=&|\int_\mathcal {R}p_n^{\xi_n(x)}(x)\mathrm{log}^2p_n(x)dx|\\
&\leq&\int_{p_n(x)\leq1}|p_n^{\xi_n(x)}(x)\mathrm{log}^2p_n(x)|dx\\
& &+\int_{p_n(x)>1}|p_n^{\xi_n(x)}(x)\mathrm{log}^2p_n(x)|dx\\
&\leq&\int_{p_n(x)\leq1}|p_n^{\frac{3}{4}}(x)\mathrm{log}^2p_n(x)|dx\\
& &+\int_{p_n(x)>1}|p_n^{\frac{3}{2}}(x)\mathrm{log}^2p_n(x)|dx\\
&\leq&\int_{\mathcal {R}}p_n^{\frac{1}{4}}(x)|p_n^{\frac{1}{2}}(x)\mathrm{log}^2p_n(x)|dx\\
& &+\int_{\mathcal
{R}}p_n(x)|p_n^{\frac{1}{2}}(x)\mathrm{log}^2p_n(x)|dx\\
&\leq&B(\int_\mathcal {R}p_n^{\frac{1}{4}}(x)dx+1),
\end{eqnarray*}
where $$B=\underset{x\in\mathcal {R \backslash
 A}, n\in\mathbb{N}}{\mathrm{sup}}|p_n^{\frac{1}{2}}(x)\mathrm{log}^2p_n(x)|<\infty$$
by condition 2) of Theorem 1. From condition 3) in Theorem 1, we
know that $\int_\mathcal {R}p_n(x)^{\frac{1}{4}}dx\leq
M_{\frac{1}{4}}$. Letting $L=B(1+M_{\frac{1}{4}})$, we obtain that
$|J_n(t)/(1-t)|\leq L$ for $n\in\mathbb{N}$.

Combining with inequality (\ref{R_S}), we have that
\begin{eqnarray*}
|R_{\alpha}(Y_n)-H(Y_n)|&\leq&\frac{1}{|1-\alpha |}
\int_{1}^{\alpha}
|t-1|\left |\frac{J_n(t)}{1-t}\right |dt\\
&\leq&\frac{1}{|1-\alpha |}\left |\int_{1}^{\alpha}
L|t-1|dt\right |\\
&=&L|1-\alpha|.
\end{eqnarray*}

Therefore, for any $\varepsilon
>0$, there exists a $\delta =\min\{\frac{\varepsilon}{L},\frac{1}{4}\}$, such that $|T_{\alpha}(Y_n)-H(Y_n))|<\varepsilon
$ uniformly for all $n\in\mathbb{N}$ and all points $\alpha$
satisfying $|\alpha-1|<\delta$ and $\frac{3}{4}\leq \alpha \leq
\frac{3}{2}$.

The conclusion 2) can be obtained by similar arguments.
\end{proof}

\vspace{0.2cm}

{\bf Proof of Lemma \ref{ineqal2}.}

\begin{proof} We just prove the case: $0\leq y\leq x\leq 1$,
the proof of the case $0\leq x\leq y\leq 1$ is similar.

Suppose that
$d_{\gamma}(x)=x^{\gamma}-(1-2\gamma)x^{\gamma}\log{x}$. Since
$d_{\gamma}(x)$ is continuous in $x\in[0,1]$, for every
$\gamma\in(0,\frac{1}{2})$,  $d_{\gamma}(x)$ is bounded in
$x\in[0,1]$. Thus, for every $\gamma\in(0,\frac{1}{2})$, there
exists a $Q_{\gamma}>0$ such that
$|d_{\gamma}(x)|\leq\frac{Q_{\gamma}}{2}$.

For $0\leq y\leq x\leq 1$, we have $x^{1-\gamma}\geq y^{1-\gamma}$
and
$$
x^{\gamma}y^{1-\gamma}\log{x}-y^{\gamma}x^{1-\gamma}\log{y}\geq
-x^{\gamma}y^{\gamma}(x^{1-2\gamma}-y^{1-2\gamma})\log{x}\geq0.
$$
It follows that
$Q_{\gamma}|x^{1-\gamma}-y^{1-\gamma}|=Q_{\gamma}(x^{1-\gamma}-y^{1-\gamma})$
and
$|x^{\gamma}y^{1-\gamma}\log{x}-y^{\gamma}x^{1-\gamma}\log{y}|=x^{\gamma}y^{1-\gamma}\log{x}-y^{\gamma}x^{1-\gamma}\log{y}$.

Denote
\begin{eqnarray*}
A_{\gamma}(x)&:=&Q_{\gamma}(x^{1-\gamma}-y^{1-\gamma})-(x^{\gamma}y^{1-\gamma}\log{x}-y^{\gamma}x^{1-\gamma}\log{y})\\
&=&x^{1-\gamma}y^{1-\gamma}(Q_{\gamma}y^{\gamma-1}-Q_{\gamma}x^{\gamma-1}\\
& & -x^{2\gamma-1}\log{x}+y^{2\gamma-1}\log{y})\\
&:=&x^{1-\gamma}y^{1-\gamma}F_{\gamma}(x),
\end{eqnarray*}
where
$$F_{\gamma}(x):=Q_{\gamma}y^{\gamma-1}-Q_{\gamma}x^{\gamma-1}-x^{2\gamma-1}\log{x}+y^{2\gamma-1}\log{y}.$$

In what follows, we show that $A_{\gamma}(x) \ge 0$ for $0 \le y \le
x \le 1$, which implies the Lemma is true for  $0 \le y \le x \le
1$.

Obviously, if $y=0$, $A_{\gamma}(x)= Q_{\gamma}x^{1-\gamma}\geq0$.
If $y>0$,  since $F_{\gamma}(y)=0$ and
\begin{eqnarray*}
F'_{\gamma}(x)
&=&(1-\gamma)Q_{\gamma}x^{\gamma-2}+(1-2\gamma)x^{2\gamma-2}\log{x}-x^{2\gamma-2}\\
&=&x^{\gamma-2}\{(1-\gamma)Q_{\gamma}-[x^{\gamma}-(1-2\gamma)x^{\gamma}\log{x}]\}\\
&=&x^{\gamma-2}\{(1-\gamma)Q_{\gamma}-d_{\gamma}(x)]\}\\
&\geq& x^{\gamma-2}(\frac{Q_{\gamma}}{2}-\frac{Q_{\gamma}}{2})=0,
\end{eqnarray*}
we obtain for every $\gamma\in(0,\frac{1}{2})$ and $y>0$,
 $F_{\gamma}(x)\geq0$ when $x\geq y$. As a result, $A_{\gamma}(x)\geq 0$.
 \end{proof}

\vspace{0.2cm}




\ifCLASSOPTIONcaptionsoff
  \newpage
\fi

\end{document}